\documentclass{article}

\usepackage{amsmath}
\usepackage{amsthm}
\usepackage{amssymb}
\usepackage{mathtools}
\usepackage{tikz}
\usepackage{fancyvrb}
\usepackage{graphicx}
\usepackage{hyperref}
\usepackage[capitalize,nameinlink]{cleveref}
\usepackage{xcolor}
\usepackage{algorithm}
\usepackage{algpseudocode}
\usepackage{float}
\usepackage{comment}
\usepackage{setspace}
\usepackage{booktabs}

\usetikzlibrary{arrows.meta,calc}

\newtheorem{theorem}{Theorem}[section]
\newtheorem{lemma}[theorem]{Lemma}
\theoremstyle{definition}
\newtheorem{definition}{Definition}[section]
\theoremstyle{remark}

\newtheorem{observation}[theorem]{Observation}
\crefname{remark}{Remark}{Remarks}
\crefname{example}{Example}{Examples}
\crefname{observation}{Observation}{Observations}
\crefname{algorithm}{Algorithm}{Algorithms}
\crefname{figure}{Figure}{Figures}
\crefname{equation}{}{}

\title{GPU-Accelerated Search for Fast Matrix Multiplication over $\mathbb{F}_2$}
\author{
  Zeke Medley
  \quad Anagha Gokul 
  \quad Quan Luu
  \quad Panagiotis Manolios\thanks{\texttt{medley.g, gokul.a, luu.qu, p.manolios@northeastern.edu}} \\[2pt]
  Northeastern University
}

\date{}

\begin{document}

\maketitle

\begin{abstract}
We present a GPU-accelerated algorithm for searching for ways to multiply matrices with few scalar multiplications. Our algorithm searches the matrix multiplication flip graph and, on an NVIDIA H200, performs over a billion search steps per second for matrix sizes up to $7\times7$, a roughly $380\times$ improvement over previous GPU-accelerated search on the tensor for $7\times7$ matrix multiplication.  We also prove that, over \(\mathbb{F}_2\), the directed matrix multiplication flip graph is strongly connected with only flip and plus edges. This removes the need for computationally expensive reduction edges in the connectivity argument used by previous work.  Using this GPU-accelerated search procedure on our simplified flip graph, we find a way to multiply $7\times 7$ matrices over $\mathbb{F}_2$ using 245 multiplications, a three-multiplication improvement over the previous record.
\end{abstract}

\section{Introduction}
\label{sec:introduction}
The vertices of the matrix multiplication flip graph are ways to
multiply matrices. Some ways require fewer multiplications and can
outperform the naive algorithm at practical matrix sizes
\cite{zhu2026falcongemm, jangda2026subcuber}, so considerable effort has been spent
exploring the flip graph to find them. Indeed, the current records for
multiplying $5\times 5$, $6\times 6$, and $7\times 7$ matrices over
$\mathbb{F}_2$ were found via flip graph search
\cite{moosbauer2025flipgraphs, perminov2025fast}.

In a sense that we make precise below, a way to multiply matrices can be represented as a multiset of triples. The central operation of flip graph search is finding sets of triples with the same value in one position. This is a matching problem CPU-oriented implementations naturally solve with mappings between values and positions~\cite{moosbauer2025flipgraphs}.

The guiding design principle of this work is to exploit the massive parallelism available in modern GPUs for flip graph search. Thus, rather than maintaining and querying large maps, we use hardware-supported matching primitives to identify sets of triples with the same value in a position, replacing hash-table lookups with hardware.

Our contributions are as follows.
\begin{enumerate}
    \item We design a GPU-friendly flip-graph search algorithm that uses warp-level matching primitives to detect flip opportunities. This outperforms previous GPU-accelerated search~\cite{perminov2025fast} by a factor of about $380$ while searching the $7\times 7$ flip graph and executes over a billion search steps per second on an NVIDIA H200 GPU for matrix sizes up to $7\times 7$.

    \item We introduce generalized multi-summand flips, which exploit linear dependencies among summands sharing a tensor factor and expose rank reductions efficiently on GPUs. On average, generalized flips discover the record $4\times 4$ scheme $1.96\times$ as fast as an implementation without them.

    \item We prove that, over $\mathbb{F}_2$, ordinary flip and plus edges alone strongly connect the directed matrix multiplication flip graph, removing the need for a third, computationally expensive reduction edge from prior work~\cite{arai2024adaptive}. We also discuss a gap in the previous connectivity proof with a counterexample in Appendix~\ref{app:gap}. The statement of the theorem is not in question, but the construction's main loop requires an implicit invariant it fails to preserve.

    \item These culminate in a way to multiply $7\times 7$ matrices over \(\mathbb{F}_2\) with $245$ multiplications, a three-multiplication improvement over the prior record~\cite{perminov2025fast}.
\end{enumerate}

\Cref{sec:preliminaries} defines the flip graph and related notation.~\Cref{sec:connectivity-proof} contains the connectivity proof for flip and plus edges over $\mathbb{F}_2$.~\Cref{sec:gpu-accelerated-search} describes our GPU-accelerated search procedure, including our use of GPU matching primitives and generalized flips.~\Cref{sec:evaluation} compares our implementation to Perminov's GPU-based search algorithm~\cite{perminov2025fast} and describes our ablation testing for generalized flips.

\section{Related Work}

The matrix multiplication flip graph was introduced by Kauers and Moosbauer~\cite{kauers2023flip}, improving on a record for $5\times 5$ matrix multiplication set by AlphaTensor~\cite{fawzi2022discovering}. Kauers and Moosbauer's construction connected schemes if they were related by a flip or reduction transform. It was shown that with these edges the flip graph is weakly connected: by treating reduction transforms as bidirectional, every scheme is reachable from every other scheme.

Arai, Ichikawa, and Hukushima~\cite{arai2024adaptive} improved the connectivity situation by introducing a third edge, plus transitions. Empirically, these help search escape local minima. Theoretically their addition was shown to make the flip graph strongly connected.\footnote{Arai et al.\ state their connectivity theorem without the qualifier \emph{strongly}; however, their flip graph is defined as a directed graph and their proof applies transformations only in the forward direction, which establishes strong connectivity.} This work resulted in another improvement of the record for $5\times 5$ matrix multiplication.

Moosbauer and Poole introduced flip graphs with symmetry~\cite{moosbauer2025flipgraphs}. A detailed account of these symmetries is outside the scope of this paper, but the practical effect is to reduce the size of the explorable flip graph at the expense of completeness, as non-symmetric schemes become unreachable. This work resulted in further improvements on the record for $5\times 5$ and $6\times 6$ matrix multiplication.

Subsequently, Wood adapted flip graph search to the commutative setting~\cite{wood2025commutative}; Khoruzhii, Gel{\ss}, and Pokutta applied flip graph search to structured matrix multiplication, for example multiplication resulting in a symmetric matrix~\cite{khoruzhii2025structured}; Kauers and Wood~\cite{kauers2025meta} introduced the meta flip graph, adding extension and projection moves between different matrix formats; and, notably, Perminov gave a GPU-accelerated flip-graph implementation~\cite{perminov2025fast}.

\section{Preliminaries}
\label{sec:preliminaries}
Throughout, we use $\uplus$ for a multiset union, $\cup$ for the standard set union, and $\setminus$ for both multiset and set difference. Matrices and tensors are over a fixed base field; from~\Cref{sec:connectivity-proof} onward we specialize to $\mathbb{F}_2$.

\subsection{Matrix Multiplication Tensors}
\begin{definition}[Matrix multiplication tensor] \label{def:mmt}
The matrix multiplication tensor for multiplying two $n \times n$ matrices\footnote{The exposition is for square matrix multiplication, but the definitions extend directly to rectangular matrix multiplication.} is the unique tensor $\mathcal{M}^{(n)}$ satisfying, for all $n \times n$ matrices $D$ and $F$, \begin{equation} \label{eq:matmul-tensor-def} (DF)_{k,\ell} = \sum_{g,h,i,j=1}^{n} \mathcal{M}^{(n)}_{g,h,i,j,k,\ell} D_{g,h} F_{i,j}. \end{equation} Equivalently, $\mathcal{M}^{(n)}$ records which products of entries of $D$ and $F$ contribute to each entry of $DF$. \end{definition}

A rank-$R$ decomposition of $\mathcal{M}^{(n)}$ is an expression \begin{equation} \label{eq:rank-r-decomposition} \mathcal{M}^{(n)} = \sum_{r=1}^{R} A^{(r)} \otimes B^{(r)} \otimes C^{(r)}, \end{equation} where each $A^{(r)},B^{(r)},C^{(r)}$ is an $n \times n$ matrix and \[ \bigl(A \otimes B \otimes C\bigr)_{g,h,i,j,k,\ell} = A_{g,h} B_{i,j} C_{k,\ell}. \]

Such a decomposition gives an algorithm for multiplying $D$ and $F$ using $R$ multiplications involving entries of the input matrices. Indeed, substituting~\eqref{eq:rank-r-decomposition} into \eqref{eq:matmul-tensor-def} gives \[ \begin{aligned} (DF)_{k,\ell} &= \sum_{g,h,i,j} \left( \sum_{r=1}^{R} A^{(r)} \otimes B^{(r)} \otimes C^{(r)} \right)_{g,h,i,j,k,\ell} D_{g,h}F_{i,j} \\ &= \sum_{r=1}^{R} C^{(r)}_{k,\ell} \left( \sum_{g,h} A^{(r)}_{g,h}D_{g,h} \right) \left( \sum_{i,j} B^{(r)}_{i,j}F_{i,j} \right). \end{aligned} \] Thus the only multiplications involving entries of $D$ and $F$ are the $R$ products \[ m_r = \left( \sum_{g,h} A^{(r)}_{g,h}D_{g,h} \right) \left( \sum_{i,j} B^{(r)}_{i,j}F_{i,j} \right), \qquad r=1,\ldots,R. \] Conversely, any way to multiply two matrices with $R$ scalar multiplications gives a rank-$R$ decomposition of $\mathcal{M}^{(n)}$. Our goal is thus to look for low-rank decompositions of $\mathcal{M}^{(n)}$.

\subsection{The matrix multiplication flip graph}
\begin{definition}[Schemes and vertices] \label{def:vertices} The vertices of the matrix multiplication flip graph for $n \times n$ matrix multiplication are multisets \[ S = \left\{ A^{(1)} \otimes B^{(1)} \otimes C^{(1)}, \ldots, A^{(|S|)} \otimes B^{(|S|)} \otimes C^{(|S|)} \right\} \] of nonzero rank-one tensors satisfying \[ \sum_{s \in S} s = \mathcal{M}^{(n)}. \] We call such a multiset a scheme, and we store each summand with a chosen factorization $(A^{(r)},B^{(r)},C^{(r)})$; over $\mathbb{F}_2$ the factorization of a nonzero rank-one tensor is unique. Following the flip-graph literature, we call the size $|S|$ the rank of the scheme, although it is the length of this particular decomposition and may exceed the tensor rank of $\mathcal{M}^{(n)}$. \end{definition}



There is a directed edge from a vertex $S$ to a vertex $S'$ in the flip graph when $S'$ is obtained from $S$ via a \textit{flip} or \textit{plus} transform.

\begin{definition}[Flip and plus transformations] \label{def:edges} Let $S$ be a scheme. If $S$ can be transformed into $S'$ by one of the following operations, then there is a directed flip/plus edge from $S$ to $S'$.

\textbf{Flip.} Flip on the $A$ position applies to two summands \( A \otimes B \otimes C\), \( A' \otimes B' \otimes C' \) with $A=A'$ and replaces them by \(A \otimes B \otimes (C+C')\), \(A' \otimes (B'-B) \otimes C'. \)
Equivalently, \[ \begin{aligned} S' = S &\setminus \left\{ A \otimes B \otimes C, A' \otimes B' \otimes C' \right\} \\ &\uplus \left\{ A \otimes B \otimes (C+C'), A' \otimes (B'-B) \otimes C' \right\}, \end{aligned} \] with any zero summands deleted.

\textbf{Plus.} Plus on the $A$ position takes two summands \(A \otimes B \otimes C\), \(A' \otimes B' \otimes C', \) and replaces the second summand by two summands (leaving the first summand unchanged): \((A'-A) \otimes B' \otimes C'\), and \(A \otimes B' \otimes C'. \)
Equivalently,
\[ \begin{aligned} S' = S &\setminus \left\{ A' \otimes B' \otimes C' \right\} \\ &\uplus \left\{ (A'-A) \otimes B' \otimes C', A \otimes B' \otimes C' \right\}, \end{aligned} \] with any zero summands deleted.

These transformations on the $B$ and $C$ positions are defined analogously.
\end{definition}

Both transformations preserve the sum of the rank-one tensors in a scheme: in each case the inserted summands sum to the summands they replace. Therefore, every vertex reached by such a transformation is again a decomposition of $\mathcal{M}^{(n)}$.

\section{Connectivity Proof}\label{sec:connectivity-proof}

In this section, we prove that flip and plus edges are sufficient for the flip graph to be strongly connected over $\mathbb{F}_2$.

\begin{theorem}[Flip-plus connectivity over $\mathbb{F}_2$] \label{thm:main} Over $\mathbb{F}_2$, for every $n \ge 2$, the directed matrix multiplication flip graph whose edges are ordinary flip and plus transformations is strongly connected. Equivalently, for any two schemes $S$ and $S'$ for $\mathcal{M}^{(n)}$, there is a directed path from $S$ to $S'$ using only flip and plus transformations. \end{theorem}

Our choice to consider connectivity over $\mathbb{F}_2$ stems from two observations. First, Moosbauer and Poole~\cite{moosbauer2025flipgraphs} do not explicitly search for reduction edges in their flip graph algorithm, noting that they occur rarely and slow down search. Second, several existing flip-graph works~\cite{kauers2023flip, arai2024adaptive, moosbauer2025flipgraphs} perform their searches over the field $\mathbb{F}_2$. Because the primary use of reduction edges in the previous connectivity proof~\cite{arai2024adaptive} is to remove zero-summing submultisets of a matrix multiplication scheme, this led us to consider whether we could exploit the properties of $\mathbb{F}_2$ for it instead. Thus, our proof relies on the following observation of $\mathbb{F}_2$.

\begin{observation}
\label[observation]{rem:delete-identical-pair}
Over $\mathbb{F}_2$, two identical summands can be deleted using a flip. If two copies of $A\otimes B\otimes C$ appear in a scheme, then a flip on the $A$ position produces \[ A\otimes B\otimes (C+C) \qquad\text{and}\qquad A\otimes (B-B)\otimes C, \] both of which are zero over $\mathbb{F}_2$ and can therefore be removed.
\end{observation}


Throughout this section we assume $n \geq 2$; this is the relevant regime for matrix multiplication search, and it avoids the degenerate one-dimensional case.

\begin{lemma}\label{lemma:mmt}
    The matrix multiplication tensor for multiplying two $n \times n$ matrices is
    \[
    \mathcal{M}^{(n)}=\sum_{i=1}^n \sum_{j=1}^n \sum_{k=1}^n E_{ij} \otimes E_{jk} \otimes E_{ik},
    \]
    where $E_{ij}$ is an $n \times n$ matrix with a $1$ at position $(i,j)$ and zeros everywhere else.
\end{lemma}
\begin{proof}
By  Definition~\ref{def:mmt}, the coefficient $\mathcal{M}^{(n)}_{g,h,i,j,k,\ell}$ is $1$ exactly when the product $D_{g,h}F_{i,j}$ contributes to $(DF)_{k,\ell}$. This happens precisely when $g=k$, $h=i$, and $j=\ell$. Hence \[ \mathcal{M}^{(n)}_{g,h,i,j,k,\ell} = \delta_{g,k}\delta_{h,i}\delta_{j,\ell}. \]

\noindent Therefore, inspecting the standard basis expansion of $\mathcal{M}^{(n)}$ gives \[ \begin{aligned} \mathcal{M}^{(n)} &= \sum_{i,j,j',k,i',k'} \delta_{i,i'}\delta_{k,k'}\delta_{j,j'} E_{i,j}\otimes E_{j',k}\otimes E_{i',k'} \\ &= \sum_{i,j,k} E_{i,j}\otimes E_{j,k}\otimes E_{i,k}. \end{aligned} \]
\end{proof}

\begin{lemma}\label{lemma:span}
Let $S$ be a decomposition of $\mathcal{M}^{(n)}$ over a field $\mathbb{F}$: \[ S = \left\{ A^{(1)}\otimes B^{(1)}\otimes C^{(1)}, \ldots, A^{(|S|)}\otimes B^{(|S|)}\otimes C^{(|S|)} \right\}. \] Then \[ \begin{aligned} \operatorname{span}_{\mathbb{F}}\{A^{(r)} : r \in [|S|]\} &= \operatorname{span}_{\mathbb{F}}\{B^{(r)} : r \in [|S|]\} \\ &= \operatorname{span}_{\mathbb{F}}\{C^{(r)} : r \in [|S|]\} = \mathbb{F}^{n\times n}. \end{aligned} \]
\end{lemma}

\begin{proof}
    An alternative proof appears as Lemma~5.1 of Arai et al.~\cite{arai2024adaptive}.

    Let $R=|S|$. We prove the claim for the $A$-components; the arguments for the $B$- and $C$-components are symmetric.

    Write each $B^{(r)}$ and $C^{(r)}$ in the standard basis: \[ B^{(r)} = \sum_{j,k'=1}^{n} \beta^{(r)}_{j,k'}E_{j,k'}, \qquad C^{(r)} = \sum_{i,k=1}^{n} \gamma^{(r)}_{i,k}E_{i,k}. \]

    By~\Cref{lemma:mmt} and the assumption that $S$ decomposes $\mathcal{M}^{(n)}$, \[ \begin{aligned} \sum_{i,j,k=1}^{n} E_{i,j}\otimes E_{j,k}\otimes E_{i,k} &= \sum_{r=1}^{R} A^{(r)}\otimes B^{(r)}\otimes C^{(r)} \\ &= \sum_{r=1}^{R} A^{(r)} \otimes \left( \sum_{j,k'=1}^{n} \beta^{(r)}_{j,k'}E_{j,k'} \right)\\ &\qquad\qquad\otimes \left( \sum_{i,k=1}^{n} \gamma^{(r)}_{i,k}E_{i,k} \right) \\ &= \sum_{j,k',i,k=1}^{n} \left( \sum_{r=1}^{R} A^{(r)} \beta^{(r)}_{j,k'} \gamma^{(r)}_{i,k} \right)\\ &\qquad\qquad\otimes E_{j,k'}\otimes E_{i,k}. \end{aligned} \]

    Now fix $i$ and $j$, and choose any $k \in [n]$. Comparing the coefficients of $E_{j,k}\otimes E_{i,k}$ in the last two tensor factors, we get

    \[ E_{i,j} = \sum_{r=1}^{R} A^{(r)} \beta^{(r)}_{j,k} \gamma^{(r)}_{i,k}. \] Thus every basis matrix $E_{i,j}$ lies in $\operatorname{span}_{\mathbb{F}}\{A^{(r)} : r\in[R]\}$, so the $A$-components span all of $\mathbb{F}^{n\times n}$.
\end{proof}

\begin{lemma} \label{lemma:add-duplicate-sum}
Let $S$ be a scheme over $\mathbb{F}_2$, and suppose \[ A\otimes B\otimes C, \qquad A'\otimes B'\otimes C' \] are summands of $S$. Let $X=A+A'$. Then there is a path from $S$ to \[ S \uplus \left\{ X\otimes B\otimes C,\; X\otimes B\otimes C \right\}.\] Analogous statements hold for the $B$ and $C$ positions. \end{lemma}

\begin{figure*}[t]
\centering
\resizebox{\linewidth}{!}{%
\begin{tikzpicture}[x=1cm,y=1cm]
\tikzset{
  term/.style={
    inner xsep=2pt,
    inner ysep=1pt,
    font=\footnotesize,
    fill=white,
    text=black
  },
  op node/.style={
    fill=white,
    inner sep=1.2pt,
    font=\scriptsize,
    text=black
  },
  gen/.style={
    -{Latex[length=1.45mm,width=0.95mm]},
    draw=black,
    line width=0.45pt
  },
  carry/.style={
    -{Latex[length=1.15mm,width=0.75mm]},
    draw=black!27,
    line width=0.30pt
  }
}

\newcommand{\tensor}[3]{#1\otimes #2\otimes #3}

\node[term] (s0abc)     at (0.0, 0.34) {$\tensor{A}{B}{C}$};
\node[term] (s0apbpcp)  at (0.0,-0.34) {$\tensor{A'}{B'}{C'}$};

\node[term] (s1apbpcp)  at (3.45, 0.68) {$\tensor{A'}{B'}{C'}$};
\node[term] (s1xbc)     at (3.45, 0.00) {$\tensor{X}{B}{C}$};
\node[term] (s1apbc)    at (3.45,-0.68) {$\tensor{A'}{B}{C}$};

\node[term] (s2apbpcp)  at (6.90, 1.02) {$\tensor{A'}{B'}{C'}$};
\node[term] (s2xbbpc)   at (6.90, 0.34) {$\tensor{X}{(B+B')}{C}$};
\node[term] (s2xbpc)    at (6.90,-0.34) {$\tensor{X}{B'}{C}$};
\node[term] (s2apbc)    at (6.90,-1.02) {$\tensor{A'}{B}{C}$};

\node[term] (s3apbpcp)  at (11.05, 1.36) {$\tensor{A'}{B'}{C'}$};
\node[term] (s3abc)     at (11.05, 0.68) {$\tensor{A}{B}{C}$};
\node[term] (s3xbbpc)   at (11.05, 0.00) {$\tensor{X}{(B+B')}{C}$};
\node[term] (s3xbpc)    at (11.05,-0.68) {$\tensor{X}{B'}{C}$};
\node[term] (s3xbc)     at (11.05,-1.36) {$\tensor{X}{B}{C}$};

\node[term] (s4apbpcp)  at (14.85, 1.02) {$\tensor{A'}{B'}{C'}$};
\node[term] (s4abc)     at (14.85, 0.34) {$\tensor{A}{B}{C}$};
\node[term] (s4xbc0)    at (14.85,-0.34) {$\tensor{X}{B}{C}$};
\node[term] (s4xbc1)    at (14.85,-1.02) {$\tensor{X}{B}{C}$};

\node[op node] (pAone) at (1.725, 0.00) {$P_A$};
\node[op node] (pB)    at (5.00,  0.34) {$P_B$};
\node[op node] (pAtwo) at (8.95,  0.00) {$P_A$};
\node[op node] (fA)    at (13.00, -0.34) {$F_A$};

\draw[carry] (pAone.east)     to[out=18,in=180] (s1apbpcp.west);
\draw[carry] (pB.east)        to[out=18,in=180] (s2apbpcp.west);
\draw[carry] (s1apbc.east)    to[out=0,in=180] (s2apbc.west);
\draw[carry] (s2apbpcp.east)  to[out=0,in=180] (s3apbpcp.west);
\draw[carry] (pAtwo.east)     -- (s3xbbpc.west);
\draw[carry] (s2xbpc.east)    to[out=0,in=180] (s3xbpc.west);
\draw[carry] (s3apbpcp.east)  to[out=0,in=180] (s4apbpcp.west);
\draw[carry] (s3abc.east)     to[out=0,in=180] (s4abc.west);
\draw[carry] (s3xbc.east)     to[out=0,in=180] (s4xbc1.west);

\draw[gen] (s0abc.east)       to[out=0,in=150] (pAone.west);
\draw[gen] (s0apbpcp.east)    to[out=0,in=210] (pAone.west);
\draw[gen] (pAone.east)       to[out=0,in=190] (s1xbc.west);
\draw[gen] (pAone.east)       to[out=-8,in=170] (s1apbc.west);

\draw[gen] (s1apbpcp.east)    to[out=0,in=150] (pB.west);
\draw[gen] (s1xbc.east)       to[out=0,in=210] (pB.west);
\draw[gen] (pB.east)          to[out=0,in=185] (s2xbbpc.west);
\draw[gen] (pB.east)          to[out=-10,in=175] (s2xbpc.west);

\draw[gen] (s2xbbpc.east)     to[out=0,in=150] (pAtwo.west);
\draw[gen] (s2apbc.east)      to[out=0,in=210] (pAtwo.west);
\draw[gen] (pAtwo.east)       to[out=14,in=190] (s3abc.west);
\draw[gen] (pAtwo.east)       to[out=-10,in=170] (s3xbc.west);

\draw[gen] (s3xbbpc.east)     to[out=0,in=150] (fA.west);
\draw[gen] (s3xbpc.east)      to[out=0,in=210] (fA.west);
\draw[gen] (fA.east)          -- (s4xbc0.west);
\end{tikzpicture}%
}
\caption{Derivation of $S_{A+}$, our name for the scheme $S \uplus \{X\otimes B\otimes C, X\otimes B\otimes C\}$ produced by~\Cref{lemma:add-duplicate-sum}. Here, $X$ denotes $A+A'$, $P_A$ and $P_B$ denote a plus on the $A$ and $B$ positions of their inputs, and $F_A$ denotes a flip on the $A$ position.}
\label{fig:s-a-plus-derivation}
\end{figure*}

\begin{proof}
If $X=0$, then the two added summands are zero and are deleted, so the trivial path suffices. Hence assume $X\neq 0$. We suppress all summands not involved in the transformations.

\Cref{fig:s-a-plus-derivation} shows a visual derivation when $B\neq B'$. In prose, starting with \[ A\otimes B\otimes C, \qquad A'\otimes B'\otimes C', \] apply a plus on the $A$ position, using $A'\otimes B'\otimes C'$ as the pivot and replacing $A\otimes B\otimes C$. This gives \[ A'\otimes B'\otimes C', \qquad X\otimes B\otimes C, \qquad A'\otimes B\otimes C. \] Next, apply a plus on the $B$ position, using $A'\otimes B'\otimes C'$ as the pivot and replacing $X\otimes B\otimes C$. This gives

\[\begin{gathered} A'\otimes B'\otimes C', \qquad X\otimes (B+B')\otimes C,\\ X\otimes B'\otimes C, \qquad A'\otimes B\otimes C. \end{gathered}\]

Now apply a plus on the $A$ position, using $X\otimes (B+B')\otimes C$ as the pivot and replacing $A'\otimes B\otimes C$. Since $A'+X=A$, this gives \[\begin{gathered} A'\otimes B'\otimes C', \qquad A\otimes B\otimes C, \qquad X\otimes (B+B')\otimes C,\\ X\otimes B'\otimes C, \qquad X\otimes B\otimes C. \end{gathered}\]

Finally, flip the two summands \[ X\otimes (B+B')\otimes C \qquad\text{and}\qquad X\otimes B'\otimes C \] on the $A$ position. This produces one zero summand and one copy of $X\otimes B\otimes C$. Thus the net effect is to restore the two original summands and add two copies of $X\otimes B\otimes C$.

It remains to handle $B=B'$. After the first plus on the $A$ position (see Figure ~\ref{fig:s-a-plus-derivation}), we have \[ A'\otimes B\otimes C', \qquad X\otimes B\otimes C, \qquad A'\otimes B\otimes C. \]

Apply a plus on the $A$ position, using $X\otimes B\otimes C$ as the pivot and replacing $A'\otimes B\otimes C$. Since $A'+X=A$, this gives \[ A'\otimes B\otimes C', \qquad X\otimes B\otimes C, \qquad A\otimes B\otimes C, \qquad X\otimes B\otimes C. \] This again restores the two original summands and adds two copies of $X\otimes B\otimes C$.
\end{proof}

\begin{lemma}\label{lem:add-arbitrary}
    Let $S$ be a scheme of $\mathcal{M}^{(n)}$ over $\mathbb{F}_2$.
    For arbitrary nonzero $X,Y,Z$, there is a path from $S$ to
    \[
        S \uplus
        \{
            X\otimes Y\otimes Z,\;
            X\otimes Y\otimes Z
        \},
    \]
    in the matrix multiplication flip graph.
\end{lemma}
\begin{proof}
We first show how to add two copies of $X\otimes B_0\otimes C_0$ for some summand $A_0\otimes B_0\otimes C_0$ of $S$.

By~\Cref{lemma:span}, the $A$-components of the summands of $S$ span $\mathbb{F}_2^{n\times n}$. Since $n\geq 2$, this space has dimension at least two. Therefore we may choose a summand $A_0\otimes B_0\otimes C_0$ with $A_0\neq X$. Extend $A_0$ to a basis $\mathcal{B}$ of $\mathbb{F}_2^{n\times n}$ using $A$-components of summands of $S$.

Write $X$ in this basis. We claim that there exist basis elements $U_1,\ldots,U_t\in \mathcal{B}$, each of which is the $A$-component of some summand of $S$, such that

\[ X=A_0+U_1+\cdots+U_t \] and every partial sum \[ P_j=A_0+U_1+\cdots+U_j, \qquad j=0,1,\ldots,t, \] is nonzero.

Indeed, if the coefficient of $A_0$ in the basis expansion of $X$ is $1$, then write \[ X=A_0+U_1+\cdots+U_t \] using the remaining basis elements appearing in the expansion. Since $A_0\neq X$, we have $t\geq 1$, and every partial sum is a nonempty sum of distinct basis elements, hence nonzero. If the coefficient of $A_0$ in the expansion of $X$ is $0$, write
\[ X=U_1+\cdots+U_m \] and use the ordered expression \[ X=A_0+U_1+\cdots+U_m+A_0. \] All proper partial sums contain $A_0$ together with some subset of the other basis elements, and are therefore nonzero; the final sum is $X\neq 0$. Now apply~\Cref{lemma:add-duplicate-sum} iteratively. Starting from the summand $A_0\otimes B_0\otimes C_0$, the first application adds two copies of \( P_1\otimes B_0\otimes C_0. \)
Using one of these copies together with a summand whose $A$-component is $U_2$, the next application adds two copies of \( P_2\otimes B_0\otimes C_0. \) Continuing in this way, we add two copies of \( P_j\otimes B_0\otimes C_0 \) for each $j=1,\ldots,t$, and in particular two copies of $X\otimes B_0\otimes C_0$. The duplicate pairs corresponding to the intermediate partial sums $P_1,\ldots,P_{t-1}$ can be deleted using~\Cref{rem:delete-identical-pair}. Thus we reach \[ S \uplus \left\{ X\otimes B_0\otimes C_0,\; X\otimes B_0\otimes C_0 \right\}. \]

We next change the $B$-component from $B_0$ to $Y$. If $B_0=Y$, there is nothing to do. Otherwise, repeat the same argument in the $B$ position, starting from one of the two copies of \(X\otimes B_0\otimes C_0\): using~\Cref{lemma:span}, choose a nonzero sequence of partial sums from $B_0$ to $Y$, apply~\Cref{lemma:add-duplicate-sum} in the $B$ position, and delete the duplicate pairs corresponding to the old component $B_0$ and to the intermediate partial sums. This leaves two copies of \( X\otimes Y\otimes C_0. \) If $C_0=Z$, we are done. Otherwise, repeat the same argument in the $C$ position, deleting the old duplicate pair \(X\otimes Y\otimes C_0\) after the pair with $C$-component $Z$ has been created. The result is \(S\uplus\{X\otimes Y\otimes Z,\;X\otimes Y\otimes Z\}\), as desired.
%
\end{proof}

\begin{lemma} \label{lemma:zero-subset} Let $S$ be a scheme for $\mathcal{M}^{(n)}$ over $\mathbb{F}_2$. Let $R$ be a submultiset of $S$ such that \[ \sum_{r\in R} r = 0. \] Then there is a path from $S$ to $S\setminus R$ using only flip and plus edges.
\end{lemma}

\begin{proof}
We show how to replace every summand in $R$ by a multiset of standard basis tensors, without changing the represented tensor. Since the summands in $R$ sum to zero, the resulting basis tensors will occur with even multiplicity and can then be deleted in identical pairs.

Let $A\otimes B\otimes C$ be a summand in the current copy of $R$. If $A$ is not a standard basis matrix, write \[ A=E_{i_1,j_1}+E_{i_2,j_2}+\cdots+E_{i_t,j_t}, \qquad t\geq 2. \] By~\Cref{lem:add-arbitrary}, add two copies of \( E_{i_1,j_1}\otimes B\otimes C. \) Use one of these two copies in a flip on the $B$ position with $A\otimes B\otimes C$. Since the two summands have the same $B$-component and the same $C$-component, this replaces them by \( (A-E_{i_1,j_1})\otimes B\otimes C \) after deleting the zero summand. The other added copy remains. Hence the net effect is to replace \[ A\otimes B\otimes C \] by \[ (A-E_{i_1,j_1})\otimes B\otimes C \qquad\text{and}\qquad E_{i_1,j_1}\otimes B\otimes C. \] Repeating this step decomposes the $A$-component into standard basis matrices. Applying the same procedure to the $B$-components and then to the $C$-components decomposes every summand of $R$ into a multiset of standard basis tensors \( E_{i,j}\otimes E_{k,\ell}\otimes E_{p,q}. \) After decomposing all summands in $R$, the resulting scheme has the form \( (S\setminus R)\uplus R', \) where $R'$ is a multiset of standard basis tensors. Since each decomposition step preserves the represented tensor, we have \[ \sum_{r\in R'} r = \sum_{r\in R} r = 0. \] The tensors \( E_{i,j}\otimes E_{k,\ell}\otimes E_{p,q} \) are distinct standard basis tensors of $(\mathbb{F}_2^{n\times n})^{\otimes 3}$, and hence are linearly independent. Therefore every standard basis tensor appearing in $R'$ appears with even multiplicity (as their coefficient must be zero over $\mathbb{F}_2$). Pair equal tensors and delete each pair using~\Cref{rem:delete-identical-pair}. This removes all of $R'$ and leaves $S\setminus R$.
\end{proof}

We are now in a position to prove the main theorem.

\begin{proof}[Proof of~\Cref{thm:main}] Let $S$ and $S'$ be two schemes for $\mathcal{M}^{(n)}$ over $\mathbb{F}_2$. Starting from $S$, use~\Cref{lem:add-arbitrary} to add two copies of every summand of $S'$. This gives a directed path from $S$ to \( S \uplus S' \uplus S'. \)
Since both $S$ and $S'$ sum to $\mathcal{M}^{(n)}$, the submultiset \( R = S \uplus S' \) of $S\uplus S'\uplus S'$ sums to \( \mathcal{M}^{(n)}+\mathcal{M}^{(n)}=0 \) over $\mathbb{F}_2$. By~\Cref{lemma:zero-subset}, there is a directed path from \( S\uplus S'\uplus S' \) to \((S\uplus S'\uplus S')\setminus (S\uplus S') = S'. \)
Thus there is a directed path from $S$ to $S'$. Since $S$ and $S'$ were arbitrary, the directed flip/plus graph is strongly connected.
\end{proof}

\section{GPU-Accelerated Search}\label{sec:gpu-accelerated-search}

This section gives an overview of our algorithm and of a generalization of flip transforms suitable for a performant in-GPU implementation.

\subsection{Search Procedure}

The most performance-sensitive part of random walks on the flip graph is identifying flip opportunities, i.e. summands that match in the $A$, $B$, or $C$ position.

\[
\overbrace{A\hspace{0.4em}\otimes \underbrace{B}_{\mathclap{\text{Term at $B$ position}}}\otimes\hspace{0.4em} C}^{\text{Summand}}
\]

Over $\mathbb{F}_2$, an $8\times 8$ term fits in a single $64$-bit word,
so existing flip graph search procedures for small matrix sizes maintain, for each position $A,B,C$, a map that stores the indices of summands with a particular value at that position. For example, if summands at indices $1,2,3$ have the value $x$ for their $A$ term and $m_A$ is the map for $A$ terms, then $m_A(x)=[1,2,3]$. For very small matrix sizes, $4\times 4$ or less, this map can be implemented with an array by interpreting a term's bit-level representation as an index, whereas larger sizes require hashing. However, this approach is not GPU-friendly, as it is branchy, requires operating on sizable buffers in shared memory, and is generally difficult to adapt to the single instruction, multiple threads model used by GPUs.\footnote{A GPU consists of streaming multiprocessors that execute threads in groups of 32, called warps, under the single-instruction-multiple-threads model: the threads (lanes) of a warp issue the same instruction on their own data. When control flow diverges within a warp, the two sides of a branch execute serially under different active masks, reducing utilization; this is why branchy code performs poorly. Warp-level primitives such as \texttt{\_\_match\_any\_sync} let the lanes of a warp exchange and compare values directly.}

Fortunately, recent CUDA GPUs have a warp-level primitive, \texttt{\_\_match\_any\_sync(mask, value)}, which takes the mask of participating lanes and a value, and returns the bitmask of participating lanes whose values equal the caller's.
For a calling thread $t$, this means
\[
\begin{aligned}
    &\bigl(\texttt{\_\_match\_any\_sync(}x_t\texttt{)}\bigr)_i=1 \\
    &\qquad\leftrightarrow\quad \text{participating thread $i$ has $x_i=x_t$}.
\end{aligned}
\]

When there are fewer summands than threads in a warp, this can be used as follows: Each thread stores a single summand. At each step, all threads agree on a term position ($A$, $B$, or $C$) and execute \texttt{\_\_match\_any\_sync} on their summand's term at that position. Ones in the resulting bitmasks correspond to flip opportunities. This means flips are identified with a hardware primitive instead of a hash table.

In practice, CUDA GPUs have $32$ threads per warp. As world-record decompositions of $\mathcal{M}^{(4)}$ and $\mathcal{M}^{(5)}$ involve $47$~\cite{fawzi2022discovering} and $93$~\cite{moosbauer2025flipgraphs} summands, respectively, multiple summands must be stored per thread. Flip search then becomes probabilistic: At each step, warps agree on a term position; then, for $k$ rounds, each thread picks a random summand and calls \texttt{\_\_match\_any\_sync}. Increasing $k$ increases the likelihood that available flip opportunities are found. Periodically, threads shuffle their summands among each other so flip opportunities between summands on the same thread are not ignored. We evaluate the effect of $k$ in~\Cref{sec:eval-k}.

\subsection{Generalized Flips}\label{sec:generalized-flips}

\texttt{\_\_match\_any\_sync} can detect multiple summands which share a term as this corresponds to more than two bits being set in the returned bitmask. Flip semantics can be generalized to operate on multiple summands and used to look ahead for rank reductions.

Sketch a flip opportunity as matrix multiplication over subterms, with $\otimes$ as multiplication.

\[
\begin{aligned}
&A\otimes (B^{(1)}\otimes C^{(1)} + B^{(2)}\otimes C^{(2)}) \\&= A\otimes \bigl(\underbrace{\begin{bmatrix}
    B^{(1)} & B^{(2)}
\end{bmatrix}}_{B}\underbrace{\begin{bmatrix}
    C^{(1)} \\ C^{(2)}
\end{bmatrix}}_{C}\bigr)
\end{aligned}
\]

\noindent From this perspective, applying a flip has the same effect as transforming $BC$, above, into $BGG^{-1}C$ for a particular $G$.

\[
\begin{aligned}
&A\otimes \bigl(\begin{bmatrix}
    B^{(1)} & B^{(2)}
\end{bmatrix} \overbrace{\begin{bmatrix}
    1 & 0 \\
    1 & 1
\end{bmatrix}}^{G}\overbrace{\begin{bmatrix}
    1 & 0 \\
    1 & 1
\end{bmatrix}}^{G^{-1}}\begin{bmatrix}
    C^{(1)} \\ C^{(2)}
\end{bmatrix}\bigr) \\
&= A\otimes(B^{(1)}+B^{(2)})\otimes C^{(1)} \\
&\qquad+A\otimes B^{(2)}\otimes(C^{(1)}+C^{(2)})
\end{aligned}
\]

This is easily generalized to multiple summands with a shared term and an arbitrary invertible $G$.

\[
\begin{aligned}
    &A\otimes (\sum_{i=1}^N B^{(i)}\otimes C^{(i)})
    \\&= A\otimes \bigl(\begin{bmatrix}
    B^{(1)} & \dots & B^{(N)}
\end{bmatrix}GG^{-1}\begin{bmatrix}
    C^{(1)} \\ \dots \\ C^{(N)}
\end{bmatrix}\bigr)
\end{aligned}
\]

A consequence of this framing is that if the $B^{(i)}$ or $C^{(i)}$ terms are not linearly independent, then an appropriate choice of $G$ results in a zero in the $B$ or $C$ vectors, which corresponds to a rank reduction. Thus, an appropriate choice of $G$ can skip ahead in the search space to find rank reductions that may require many pairwise flip and plus transforms to set up the needed linear combination.

It remains to discuss how to efficiently determine on a GPU whether $B^{(1)},\dots,B^{(N)}$ are linearly dependent. Sets of summands sharing a term appear to be of size $\leq 5$ on the $5\times 5$ flip graph: among $251{,}178$ sets of more than two summands sharing a term in $16{,}000$ schemes recorded along descents from the naive scheme, all had size $\leq 5$;~\Cref{fig:sizedist} visualizes this. Over $\mathbb{F}_2$, determining whether $B^{(1)},\dots,B^{(5)}$ are linearly dependent reduces to determining whether there is a nonempty subset whose sum (i.e. XOR, i.e. $\oplus$) is zero. This makes $5$ a natural design threshold: sets of size $5$ have $31$ nonempty subsets, and warps have $32$ threads, so we can assign each subset to a thread. Larger groups are not enumerated; when a matching class has more than five members, the step falls back to an ordinary pairwise flip within the class.~\Cref{fig:groupsupply} reports the corresponding group-size distributions for sizes $4\times4$ through $7\times7$, measured on schemes recorded along descents; groups within the threshold account for the large majority at every size.

\begin{figure}[h]
    \centering
\begin{tikzpicture}[x=0.62cm,y=0.55cm]
  \draw[black,thin] (0,0) -- (0,6.55);

  \foreach \y/\lab in {0/$10^0$,1/$10^1$,2/$10^2$,3/$10^3$,4/$10^4$,5/$10^5$,6/$10^6$} {
    \draw[black,thin] (0,\y) -- (-0.10,\y);
    \node[anchor=east,font=\small] at (-0.16,\y) {\lab};
  }

  \draw[fill=white,draw=black,thin] (0,0) rectangle (1,5.252);
  \draw[fill=white,draw=black,thin] (1,0) rectangle (2,4.741);
  \draw[fill=white,draw=black,thin] (2,0) rectangle (3,4.241);

  \foreach \x/\lab in {0.5/3,1.5/4,2.5/5} {
    \node[anchor=north,font=\small] at (\x,-0.08) {\lab};
  }
  \node[anchor=north,font=\small] at (1.5,-0.85) {matching-group size};
  \node[anchor=south,rotate=90,font=\small] at (-1.35,3.27) {observed groups};
\end{tikzpicture}
    \caption{Distribution of sizes of sets of summands sharing a term in $16{,}000$ schemes recorded along descents from the naive scheme on the $5\times 5$ flip graph.}
    \label{fig:sizedist}
\end{figure}

\begin{figure}[t]
\centering
\footnotesize
\setlength{\tabcolsep}{2pt}
\begin{tabular*}{\columnwidth}{@{\extracolsep{\fill}}lccc}
\toprule
Size & groups of size $3$--$5$ & size $6$--$8$ & size $\ge 9$ \\
\midrule
$4\times4$ & 1.1--7.9 & 0 & 0 \\
$5\times5$ & 11--20 & 0 & 0 \\
$6\times6$ & 22--34 & 0.3--2.6 & 0 \\
$7\times7$ & 32--49 & 3.4--10.4 & 0 \\
\bottomrule
\end{tabular*}
\caption{Average number of shared-factor groups per scheme by group size, measured on schemes recorded along descents from the naive scheme (two runs per size; ranges span descent phases from early to late). Groups of size at most five, which the warp enumerates exhaustively, account for the large majority at every size; larger groups fall back to ordinary pairwise flips.}
\label{fig:groupsupply}
\end{figure}

More precisely, for a group of $N\leq 5$ summands and each thread index $0\leq s<2^N$, let $D_s\subseteq\{1,\dots,N\}$ be the set of one-based positions of the set bits of $s$. Threads $s$ with $1\leq s<2^N$ compute
\[
\bigoplus_{i\in D_s}B^{(i)}=0.
\]
A warp ballot then finds the lowest-index thread for which the above is true. If one exists, call it thread $m$, and let $r=\min D_m$ be the least index in its selected subset. The threads rewrite their summands from
\[
\sum_{i=1}^{N} A\otimes B^{(i)}\otimes C^{(i)}
\quad\text{to}\quad
\sum_{\substack{i=1\\i\neq r}}^{N} A\otimes B^{(i)}\otimes C'^{(i)}
\]
where
\[
C'^{(i)}=\begin{cases}
    C^{(i)}\oplus C^{(r)} \quad &i\in D_m, i\neq r \\
    C^{(i)} \quad &i\notin D_m
\end{cases}
\]
and the summand indexed by $r$ is removed, which is implemented by setting it to zero.

To verify correctness,
\[
\begin{aligned}
    \sum_{\substack{i=1\\i\neq r}}^{N} A\otimes B^{(i)}\otimes C'^{(i)}
    &= \left(\sum_{\substack{i=1\\i\neq r}}^{N} A\otimes B^{(i)}\otimes C^{(i)}\right)\\ &\qquad\oplus \sum_{\substack{i\in D_m\\i\neq r}} A\otimes B^{(i)}\otimes C^{(r)} \\
    &= \left(\sum_{\substack{i=1\\i\neq r}}^{N} A\otimes B^{(i)}\otimes C^{(i)}\right)\\ &\qquad\oplus A\otimes B^{(r)}\otimes C^{(r)} \\
    &= \sum_{i=1}^{N} A\otimes B^{(i)}\otimes C^{(i)}.
\end{aligned}
\]

Note that generalized flips do not move the search outside the space of schemes: a generalized flip yields another decomposition of $\mathcal{M}^{(n)}$, and by~\Cref{thm:main} any such decomposition is reachable from the previous scheme using ordinary flip and plus transformations alone. Whether the transformation applied by a generalized flip can always be realized as a short path of ordinary flips, and at what length, is a question we plan to study more closely~(\Cref{sec:future-work}).

For a theoretical analysis of the optimal rank reduction attainable by a generalized flip, see Appendix~\ref{app:rank-reduction}.

\section{Evaluation}
\label{sec:evaluation}

\subsection{Comparison with Prior Implementations}\label{sec:eval-comparison}



As benchmarks, we compare our implementation to Perminov~\cite{perminov2025fast}'s search procedure, which is the only other known GPU algorithm for flip graph search at the time of writing, and to the original CPU implementation of Kauers and Moosbauer~\cite{kauers2023flip} (their updated code, which incorporates the plus transition of the adaptive flip graph method~\cite{arai2024adaptive}), measured on a single core of an Intel Xeon Gold 6548Y+ CPU. Throughput is reported in search steps per second, where one step applies one elementary transition (a flip, a plus, or a group reduction) to one walk. Our kernels are compiled with \texttt{nvcc -O3} for \texttt{sm\_90}; the Kauers-Moosbauer baseline is compiled with \texttt{g++ -O3} from the authors' repository. Our implementation executes between 6 and 1{,}100 times as many search steps per second as Perminov's implementation for matrix sizes between \(3\times 3\) and \(8\times 8\), as shown in~\Cref{fig:throughput}.

\begin{figure}[t]
    \centering
    \footnotesize
    \setlength{\tabcolsep}{2.5pt}
    \begin{tabular}{l|r|r|r|r|r}
      Size & Ours & Perminov~\cite{perminov2025fast} & Ratio & KM~\cite{kauers2023flip} & Ratio \\
      \hline
      $3\times3$ & $1.22$\,G & $200$\,M & ${\sim}6\times$ & $12$\,M & ${\sim}100\times$ \\
      $4\times4$ & $1.25$\,G & $60.3$\,M & ${\sim}21\times$ & $7.0$\,M & ${\sim}180\times$ \\
      $5\times5$ & $1.10$\,G & $31.9$\,M & ${\sim}34\times$ & $4.3$\,M & ${\sim}250\times$ \\
      $6\times6$ & $1.11$\,G & $9.50$\,M & ${\sim}117\times$ & $3.0$\,M & ${\sim}370\times$ \\
      $7\times7$ & $1.03$\,G & $2.69$\,M & ${\sim}385\times$ & $1.7$\,M & ${\sim}620\times$ \\
      $8\times8$ & $820$\,M & $751$\,K & ${\sim}1{,}090\times$ & $1.0$\,M & ${\sim}820\times$ \\
    \end{tabular}
    \caption{Flip-graph search throughput in search steps per second: ours vs.\ Perminov's FlipGraphGPU~\cite{perminov2025fast} and the Kauers-Moosbauer CPU implementation~\cite{kauers2023flip} (single core), all over $\mathbb{F}_2$ from the trivial rank-$n^3$ scheme. Each ratio column compares ours to the implementation on its left. GPU measurements on an NVIDIA H200; Kauers-Moosbauer on one core of an Intel Xeon Gold 6548Y+.}
    \label{fig:throughput}
\end{figure}

\begin{figure*}[t]
\centering
\begin{minipage}[b]{0.47\textwidth}
\centering
\begin{tikzpicture}[x=0.44cm,y=0.28cm]
  \draw[black,thin] (0,0) -- (0,23.8);
  \foreach \y in {0,5,10,15,20} {
    \draw[black,thin] (0,\y) -- (-0.20,\y);
    \node[anchor=east,font=\small] at (-0.30,\y) {\y};
  }

  \draw[fill=blue!75,draw=black,thin] (0.5,0)  rectangle (1.5,16.4);
  \draw[fill=blue!25,draw=black,thin] (0.5,16.4)  rectangle (1.5,17);
  \draw[fill=red!75,draw=black,thin] (1.5,0)  rectangle (2.5,15.8);
  \draw[fill=red!25,draw=black,thin] (1.5,15.8)  rectangle (2.5,17);
  \draw[fill=gray!60,draw=black,thin] (2.5,0)  rectangle (3.5,11.6);
  \draw[fill=gray!25,draw=black,thin] (2.5,11.6)  rectangle (3.5,13);
  \draw[fill=blue!75,draw=black,thin] (4.5,0)  rectangle (5.5,16.8);
  \draw[fill=blue!25,draw=black,thin] (4.5,16.8)  rectangle (5.5,19);
  \draw[fill=red!75,draw=black,thin] (5.5,0)  rectangle (6.5,5.7);
  \draw[fill=red!25,draw=black,thin] (5.5,5.7)  rectangle (6.5,7);
  \draw[fill=gray!60,draw=black,thin] (6.5,0)  rectangle (7.5,16.7);
  \draw[fill=gray!25,draw=black,thin] (6.5,16.7)  rectangle (7.5,18);
  \draw[fill=blue!75,draw=black,thin] (8.5,0)  rectangle (9.5,10.5);
  \draw[fill=blue!25,draw=black,thin] (8.5,10.5)  rectangle (9.5,11);
  \draw[fill=red!75,draw=black,thin] (9.5,0)  rectangle (10.5,1.9);
  \draw[fill=red!25,draw=black,thin] (9.5,1.9)  rectangle (10.5,2);
  \draw[fill=gray!60,draw=black,thin] (10.5,0)  rectangle (11.5,0.5);
  \draw[fill=gray!25,draw=black,thin] (10.5,0.5)  rectangle (11.5,1);
  \draw[fill=blue!75,draw=black,thin] (12.5,0)  rectangle (13.5,7.7);
  \draw[fill=blue!25,draw=black,thin] (12.5,7.7)  rectangle (13.5,10);
  \draw[fill=red!75,draw=black,thin] (13.5,0)  rectangle (14.5,1);
  \draw[fill=gray!60,draw=black,thin] (14.5,0)  rectangle (15.5,0.6);
  \draw[fill=gray!25,draw=black,thin] (14.5,0.6)  rectangle (15.5,3);
  \foreach \x/\lab in {2/4, 6/5, 10/6, 14/7} {
    \node[anchor=north,font=\small] at (\x,-0.4) {$\lab\times\lab$};
  }

  \draw[fill=blue!75,draw=black,thin] (11.6,22.6) rectangle (12.4,23.4);
  \node[anchor=west,font=\small] at (12.6,23.0) {Ours};
  \draw[fill=red!75, draw=black,thin] (11.6,20.8) rectangle (12.4,21.6);
  \node[anchor=west,font=\small] at (12.6,21.2) {Perminov~\cite{perminov2025fast}};
  \draw[fill=gray!60,draw=black,thin] (11.6,19.0) rectangle (12.4,19.8);
  \node[anchor=west,font=\small] at (12.6,19.4) {KM~\cite{kauers2023flip}};
\end{tikzpicture}

{\small (a)}
\end{minipage}\hfill
\begin{minipage}[b]{0.50\textwidth}
\centering
\includegraphics[width=\linewidth]{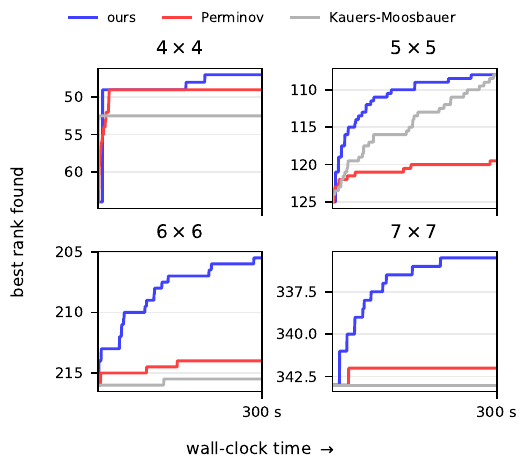}

{\small (b)}
\end{minipage}
\caption{Search progress from the naive rank-$n^3$ scheme, over ten independent 300\,s runs per implementation and size. \textbf{(a)}~Rank reduction below the trivial scheme (higher is better). Solid bars are the mean, lighter extensions the best of the ten runs. \textbf{(b)}~Median best rank found versus wall-clock time, with the vertical axis inverted so that better (lower rank) is higher; each curve's endpoint corresponds to the bar in~(a) (bars: means; curves: medians). In both panels ours runs with generalized flips disabled; they are evaluated in~\Cref{sec:eval-gf}. Kauers-Moosbauer runs as a 15-process ensemble on Intel Xeon Gold 6548Y+ cores; both GPU implementations use the H200.}
\label{fig:rank-reduction}
\end{figure*}

\Cref{fig:rank-reduction}(a) shows how this throughput translates into search progress within a fixed time budget. Perminov's implementation is built in its $\mathbb{Z}_2$ mode with size-changing transitions disabled, so all three searches explore the same space of schemes; their transition sets differ, however: ours applies flips and plus transitions (generalized flips are enabled only for the ablation of~\Cref{sec:eval-gf}), Kauers-Moosbauer applies flips and plus transitions, and Perminov's kernel applies the moves of~\cite{perminov2025fast} with the size-changing ones disabled. At $4\times4$ and $5\times5$ every implementation makes progress, though only ours reliably reaches the record rank-47 scheme at $4\times4$ (in eight of ten runs; with generalized flips enabled it does so in all ten,~\Cref{sec:eval-gf}); at $6\times6$ and $7\times7$ the prior implementations remain within 3 ranks of the trivial scheme even in their best runs, while ours removes 10 and 8 ranks on average.

The ensemble design also supports a comparison at equal hardware cost rather than equal device count. Because each Kauers-Moosbauer process is an independent single-core walk, an $N$-core ensemble returns exactly the best of $N$ independent runs, so the expected result for any $N$ follows from the empirical distribution of single-core outcomes by order statistics. From the 150 independent single-core 300\,s runs underlying~\Cref{fig:rank-reduction}(a), an ensemble sized to match the cost of one H200 at on-demand cloud prices (roughly 88 vCPUs) is expected to reach ranks $51.6$, $107.1$, $215.0$, and $341.2$ at sizes $4\times4$ through $7\times7$, versus our single-GPU measured means of $47.6$, $108.2$, $205.5$, and $335.3$ in the same 300\,s.

\Cref{fig:rank-reduction}(b) shows the full search trajectories: for each implementation and size, the median best rank found so far as a function of wall-clock time, over the same ten runs. Our implementation reaches lower ranks than both prior methods at every size.

\subsection{Generalized Flips}\label{sec:eval-gf}

\Cref{fig:rank47-h200-generalized-flips} ablates generalized flips. For discovering the current world-record scheme~\cite{fawzi2022discovering} for $4\times 4$ matrix multiplication, generalized flips result in a $1.96\times$ speedup on average. We report this ablation at $4\times 4$ only: the rate at which walks encounter groups with a usable linear dependency falls sharply with the matrix size (in our measurements, by roughly three orders of magnitude between $4\times 4$ and $7\times 7$), so at larger sizes generalized flips fire too rarely to affect wall-clock search time. Making generalized flips productive at larger sizes, including computing the optimal reduction of~\Cref{thm:opt-gen-flip-reduction}, is future work~(\Cref{sec:future-work}).

\begin{figure}[t]
\centering
\footnotesize
\setlength{\tabcolsep}{2pt}
\begin{tabular*}{\columnwidth}{@{\extracolsep{\fill}}lrrrrr}
\toprule
Search variant
  & Mean
  & Median
  & Min
  & Max
  & Hits \\
\midrule
Generalized flips
  & 39.64
  & 43.06
  & 7.57
  & 71.16
  & 10/10 \\
No Generalized flips
  & 77.65
  & 66.21
  & 41.83
  & 129.90
  & 10/10 \\
\bottomrule
\end{tabular*}
\caption{
Wall-clock time to discover a rank-47 decomposition of the
$4 \times 4 \times 4$ matrix multiplication tensor on one NVIDIA H200. All
runs started from the naive rank-64 decomposition, and all times are in seconds. The aggregate columns
summarize ten independent runs for each variant. Generalized flips result
in a mean $1.96\times$ speedup.
}
\label{fig:rank47-h200-generalized-flips}
\end{figure}

\subsection{The Probe-Round Parameter $k$}\label{sec:eval-k}

The parameter $k$ controls how persistently a step searches for a transition: each step attempts up to $3k$ matching rounds ($k$ per tensor position) and uses the first candidate found. One might expect $k$ to trade throughput against the probability of finding available edges. Measured, it does not: because a step leaves its probe loop on the first hit and matches are plentiful, steps per second are flat in $k$ (within about $3\%$ between $k=1$ and $k=8$ at every size we measured, e.g., $1.08$ versus $1.05$ billion steps per second at $7\times7$). Search quality, however, depends on $k$ strongly and non-monotonically.~\Cref{fig:kstudy} reports wall-clock time to reach a fixed target rank: $k=1$ fails on every task, and the best $k$ varies with size and task. Since a plus transition is attempted only when every round misses, $k$ indirectly sets the walk's flip-to-plus balance, and the best balance is size- and regime-dependent; per-size tuning of $k$ is therefore necessary, and the configurations of~\Cref{sec:eval-comparison} use the tuned values.

\begin{figure}[t]
\centering
\footnotesize
\setlength{\tabcolsep}{2pt}
\begin{tabular*}{\columnwidth}{@{\extracolsep{\fill}}lrrrrr}
\toprule
Task & $k{=}1$ & $k{=}2$ & $k{=}4$ & $k{=}6$ & $k{=}8$ \\
\midrule
$4\times4$: naive $\to$ 47 & -- & 2.6 & 8.6 & 51.4 & 17.4$^*$ \\
$5\times5$: naive $\to$ 112 & -- & -- & 64.2$^*$ & 46.6 & 37.0 \\
$6\times6$: naive $\to$ 210 & -- & 45.3$^*$ & 65.6 & 102 & 102 \\
$7\times7$: rank 248 $\to$ 246 & -- & -- & 20.6$^*$ & 11.1 & 16.3 \\
\bottomrule
\end{tabular*}
\caption{Mean wall-clock seconds to reach a fixed target rank as a function of the probe-round parameter $k$ (three to five runs per cell). A dash means no run reached the target within the step budget; $^*$ marks cells where only some runs reached it. The first three tasks descend from the naive scheme; the last pushes the $7\times7$ frontier from a corpus of rank-248 schemes.}
\label{fig:kstudy}
\end{figure}

\subsection{Provenance of the Rank-245 Scheme}\label{sec:eval-provenance}

The rank-245 scheme was found by a fully automated deployment of the search described in this paper, which ran unattended and continuously reseeded its walks from the pool of best schemes found so far. The deployment predates the measurement harness used in this section, and its exact configuration, seeds, and timing were not recorded; we therefore make no quantitative claim about the cost of the discovery, and we treat the scheme itself, which is independently verifiable from Appendix~\ref{sec:the-scheme}, as the result.

\section{Conclusion}
\label{sec:conclusion}

We presented a GPU-accelerated procedure for searching the matrix multiplication flip graph over $\mathbb{F}_2$. The implementation replaces global data structures for detecting flip opportunities with warp-local matching operations, making random walks on the flip graph well suited to modern GPUs. We also introduced generalized multi-summand flips, which act as GPU-efficient macro-transformations exposing local rank reductions through linear dependencies among summands.

On the theoretical side, we proved that over $\mathbb{F}_2$ the directed matrix multiplication flip graph is strongly connected using only ordinary flip and plus transformations. This removes reduction edges from the connectivity argument in the binary setting.

Experimentally, our implementation executes over a billion search steps per second on an NVIDIA H200 GPU for matrix sizes up to $7\times 7$ and substantially improves throughput over previous GPU-accelerated flip-graph search. Using this search procedure, we found a rank 245 decomposition of the $7\times 7\times 7$ matrix multiplication tensor, improving the previous best known rank by three multiplications. The scheme itself, together with a program that verifies it, is given in Appendix~\ref{sec:the-scheme}.

\section{Future Work}\label{sec:future-work}
Five additional directions seem promising to us.
First, our packed representation supports factors with up to $64$ entries, i.e., square sizes through $8\times8$ and all experiments are over $\mathbb{F}_2$. Possible extensions of this work could be to explore the flip graphs at larger sizes and using arbitrary fields.
Second, extending our GPU implementation to search the symmetric flip graphs of Moosbauer and Poole~\cite{moosbauer2025flipgraphs}, which reach lower ranks at $5\times5$ and $6\times6$ by restricting search to symmetric schemes. Third, computing optimal generalized flips on the GPU:~\Cref{thm:opt-gen-flip-reduction} characterizes the largest rank reduction a generalized flip can achieve, but our implementation does not attempt to find this optimum. A related theoretical question is which generalized flips can be realized as short paths of ordinary flips. Fourth, combining our search with systems such as FalconGEMM~\cite{zhu2026falcongemm} to search directly for schemes that perform well on particular hardware, rather than for low rank alone. Finally, extending our connectivity proof to flip graphs for structured matrix multiplication~\cite{khoruzhii2025structured}.

\section*{Acknowledgments}
This work was partially supported by the Defense Advanced Research Projects Agency under award number HR0011259E006. Approved for public release: distribution is unlimited. The views, opinions and/or findings expressed are those of the authors and should not be interpreted as representing the official views or policies
of the Department of Defense or the U.S. Government.

\bibliographystyle{plain}
\bibliography{references}

\begin{thebibliography}{10}

\bibitem{arai2024adaptive}
Yamato Arai, Yuma Ichikawa, and Koji Hukushima.
\newblock Adaptive flip graph algorithm for matrix multiplication.
\newblock In {\em Proceedings of the 2024 International Symposium on Symbolic
  and Algebraic Computation}, ISSAC '24, pages 292--298, New York, NY, USA,
  2024. Association for Computing Machinery.

\bibitem{fawzi2022discovering}
Alhussein Fawzi, Matej Balog, Aja Huang, Thomas Hubert, Bernardino
  Romera-Paredes, Mohammadamin Barekatain, Alexander Novikov, Francisco J.~R.
  Ruiz, Julian Schrittwieser, Grzegorz Swirszcz, David Silver, Demis Hassabis,
  and Pushmeet Kohli.
\newblock Discovering faster matrix multiplication algorithms with
  reinforcement learning.
\newblock {\em Nature}, 610(7930):47--53, 2022.

\bibitem{jangda2026subcuber}
Abhinav Jangda.
\newblock Compiling {Strassen}-like matrix multiplication algorithms to fast
  {CUDA} kernels.
\newblock {\em Proc. ACM Program. Lang.}, 10(PLDI), June 2026.

\bibitem{kauers2023flip}
Manuel Kauers and Jakob Moosbauer.
\newblock Flip graphs for matrix multiplication.
\newblock In {\em Proceedings of the 2023 International Symposium on Symbolic
  and Algebraic Computation}, ISSAC '23, pages 381--388, New York, NY, USA,
  2023. Association for Computing Machinery.

\bibitem{kauers2025meta}
Manuel Kauers and Isaac Wood.
\newblock Exploring the meta flip graph for matrix multiplication, 2025.
\newblock arXiv preprint arXiv:2510.19787.

\bibitem{khoruzhii2025structured}
Kirill Khoruzhii, Patrick Gel{\ss}, and Sebastian Pokutta.
\newblock Faster algorithms for structured matrix multiplication via flip graph
  search, 2025.
\newblock arXiv preprint arXiv:2511.10786.

\bibitem{moosbauer2025flipgraphs}
Jakob Moosbauer and Michael~J. Poole.
\newblock Flip graphs with symmetry and new matrix multiplication schemes.
\newblock In {\em Proceedings of the 2025 International Symposium on Symbolic
  and Algebraic Computation}, ISSAC '25, pages 233--239, New York, NY, USA,
  2025. Association for Computing Machinery.

\bibitem{perminov2025fast}
A.~I. Perminov.
\newblock Fast matrix multiplication via ternary meta flip graphs, 2025.
\newblock arXiv preprint arXiv:2511.20317.

\bibitem{wood2025commutative}
Isaac Wood.
\newblock Exploring commutative matrix multiplication schemes via flip graphs,
  2025.
\newblock arXiv preprint arXiv:2506.22113.

\bibitem{zhu2026falcongemm}
Honglin Zhu, Jiaping Cao, Jiang Shao, Siyuan Feng, Qian Qiu, Peng Chen,
  Xu~Zhang, Yixian Zhou, Man~Lung Yiu, Guang Ji, Minwen Deng, Jintao Meng, and
  Wenxi Zhu.
\newblock {FalconGEMM}: Surpassing hardware peaks with lower-complexity matrix
  multiplication, 2026.

\end{thebibliography}

\appendix
\clearpage
\section{The Scheme}\label{sec:the-scheme}

This block encodes our rank-245 decomposition of the
$7\times 7\times 7$ matrix multiplication tensor over $\mathbb F_2$.
For each rank-one term, its three $7\times 7$ binary coefficient matrices
are stored as seven little-endian bytes each, in $A,B,C$ order.
Within each matrix mask, bit $7i+j$ records coordinate $[i,j]$.
The resulting byte string is compressed with zlib and encoded with Base32.

\begin{Verbatim}[fontsize=\tiny]
BEGIN-RANK245-7X7X7-ZLIB-BASE32
PDNF2WCNRAOMOFL6KXJTHWY2QZG3PPJWVM4IPXTGBHBSFVQTVWBX3UFBNXTDAJWEN2GBO4Q4IIYE
X4CBW7WMDB22NGS3J4KCWQ4TTZFSZYQ5RYEDRQQSJAOAFBPIG3DWSEB5COPSZUDUXFQJEQHHDCG7
N5226B3JM75JWV7VL26XV5N5V7DIHPCSQAAHAHUKL45X5ECXIF77577AYEZHIQMPP4XI6RYEMYBO
L6AJY7KZB2CT46WV6MK6IJ6AKHSHWNKAMJPU7F6XCMPNQAABGBFP6IMP5LR4ESR6AHMMGPMPPYQB
DS6R3S7QEYI5GMWQOMZ6RBM5BGSH2BMPQ3SHMLFLRFF2VEUME6CHMSPE7K5ZB6CQCUI2YN77GD6H
B3SCAJXBMF2IMA5RSTVNV4NAP3JDQB4DKYPESPSHD6HR3Q2RQP4C2OWQRNWDTVSOJDLDXA4RXRP4
ZDRMH6PZQBOXQWVOLPMGUA6GSJLETQGEFV6SBFSJM5H5DJFHDND6HUT6F7KBHSBLJWJTUSRI56LS
OI4LPYLIPZYLAPKRDNDJ6KQVYA246EWXYRVLGSFXLA47TSJSYCQYGCEGPADWFKF4LAOD3HTG6I4Z
OGC4CBIOD6AZPNPH2VRQV4XWKVE6DJ4Q5ZZWDJT6D34EH3VQ7S6GJ27OP5D7KX7OLDW4FHSMWMCD
AZ2AKIYG7KY3S4OZ4T22WDKNEC6POZJL754WATBRXEF2FOZFPP6TGLTTKXB4RTFQIW3GDL7GUTZ4
ONCVK2DANYYSJMSVJ3ZLNTG7PW4ZZZKOYEJL4ZGQPUFFSBUN4TK4GGWPTLTYBFSYU7RGHIJAF5J3
GC3TZH23GZ2ADCYL6VPW3YHQEHTRIK3MFOJTEAPZTCM5C3H76OM3HZ6IUMYYHDYHHSPRPNBWFAZU
EEICHP6XUWEBKGOOXEX4H6IOUPQHA6L2JNVG2XMHION4EFEKTZ4JOA76D7TKL2UHHWGYAQ3SCSJK
NNL6MYCNOPRQB4SZOOMDVVNQKCPPSTG6DTEOQ6LM3DF7UGONL6G645RPGCPRUCAEYQ2GLUS4ZTMS
BOGB37L3MNXJNST4JHT6BPZXZRJJP5HOGFKWJEMZWJKUQ6UBIHSSYSNEGJFAVPCVTRZN4I2KJCZ2
CICXDAPH4NPGKW5TSFLGXF4KCL3LMVII43ZJOGTMDGZFD66OPVXAAJOIMFIGC46I7NZUVVFCNMNE
BJED7FBNTAUMQRWIXOQ2MOE7RCLPQ5DXHHKPOIO5PIRFUIXYTCTG4YJB6E3QKOP2IV2TO6UHJYT4
SWV232XJFOR7GV5NDN7NZL5FDVV6SES2MXWXCDRGIN243QMOLXNL2J6REA4ZEOXYENNJZKNJJACX
HN6FD2E4B67XRUABWVCODF25K2YDO3ETPJCNURLJH6SVRKLDGQHM4QLREOPECBTRNYBJ2OHVWGNZ
7C2BCRVASZ74EVTZV3A3WIRAHJ5W7YQJHDLA5WLANVLWZF5LYMV462OAWI3QW7NZYD4YCUEVJM2M
LKSFVTREZEMNXNLJV6PD2ZI6WR2VFU4UOZPB2CBYEFYALU7DDKWPMJQWORQHVAJGGUZQFKMZVKS6
BINCGZYUIE4GZXOH35UIJOJW7A45FYFMLHVWXKOYT3QAFUBA7F43N3OXXYNJD2OTQV2HSZBFAGJK
VXWRRMVFRSTNFR3NDDVLNIC4DH3OCCGBNXVTFUWH7JZQFEKUNLTK3ARKNPOJ7756MYHFDUDWNQ3R
MCBKRAFEG2UQ6N3GHU2RLFASJMA3ROPSYMNEL4W4CQYEW5Y5U7VMAD2E4YF7EPCFB4Z5BH75VZPF
TOSVVYYGZTFJXRBPMERYPNY7OEQYEH5BEMBLCQVYLKSSCKOILUKETVLVJAKWM6UMAUGN4TQQSUDW
435HIKKDJOJHYC34UZS5KVQOK5HYWRSBBB5XQ6H2WCZXMAN7FOWPBYVR37TOCSIYJ3FTEWAWRKJA
4HDRWELMY6ZH5465OOC4KSN3PFL6ELGAPAB3PQ354U54VCNW4MPU7QWD76EWLG354ZHRK3DLKJCB
PWC2KXOIYJ6V6LHHBRCW2RT3DIZE6BFDS77F7MPCEG3SWYYLFIHNA7LFCRHBIFEGSKYGMBKU5SUK
YA4BYFF4QIQPK54ZSP5DWSSNSKFSINL74I5HSDRIE5JSHDS6QBEQJS7Z6URINVA4TOS4J5RIC6LY
AO25P3DS6RYEWIDY6B4WYH2GB7IKQKSC4BR2E5XOV2XI7JWGIKRFNS3BXE2BPFRYG444RXRNIGPL
J674OIH6BNUISCVWQBHVASYKPMYTTEBOFUS4S5AIKGMLVSQRZAZHKYJHCOKSVANEKQCSJJOKKHRY
5OM7RLIF25WZPUBYX7HHPX2Y53JOYVUDVFVXKTIEDSAKSIAUHWDTPJLFFO4UFS7BGOK5RTD2YUNY
WOWS4IZLUQSPTRZGI4ELJZEMLFCITZNGJC4EAQR6KNMTQSTQBWUAIGCE3GNTSBSFDYOBASCCQLAX
LFAHFQO4PFS5BTBHX4VSMAINK7LWZKQPJ53X5XDMAYDJK6OASPTNXBXEFJD2VDDJ2TMNILRIASYK
QUCIGLBEPK6TAQJQM6SWVTOT54E7QZ2OTU2CEPC4A4CFEKC4CZ2B54D6CRTBHOY76GEZXEGM3ZRX
BYNMDBZ7S5GJDIV6STCKOKCUV7ASVWRZA7ECRHJEXHDAAEKX7CBAROXILNPLUXVBNA2WMD6EFOPV
GWB2BFWWZPSX4OHZ5WUVCNLHU6BD6JVD6CRYKMIALXXAJZGLTVUBISU6LYX565ER5G2R3AWURVKF
QTWX4EUQ3D6EEKBKBYJXEEAU2NA76XTHW4H5WZVYKLJI7MZ5WVADA47RHS4GFUG2YVGVCAUBPW4B
LVWPDP5WMQXI2KJ6EEESZDDGTJG2MC7CD667CZSKIEL4J6YJWLG5CONRBSPV25H4TMAIHLZLFD6T
24ZH6H3QE4OEMR2WFCIF7WWQD766H3Q4NRYXFVDUFIQFL5PRO4AHOX5756AKQL2F3UJZTCMPEUZ3
PLDZXYPZZ44YGAYC7WIYCO6URNNRACAUEGH5SFKQDMFFUFSYXMPZ7GV6UPV6UDSXZH23HYFLFCPP
6ON27T2B2BWPCL2YW5Z2LJK2JH3VRPBL6AY5WZCGMLHBA5VV5MUMS5IHED5Y5DBT6DNT36MNBY43
PZHZ44GSQV3ZTNNTFZUHWJETPEIVZSAFIS7U2RUOTB7NCC25QFIQMNDEFFOJP4OPBIQ42BNEZSY2
LXVAPG4FZFVT7ABUIQK3WUFJ3K2BTEZYRQKL4ROHZZLBASMS4BQZQMUQ6J7R4L2M5TE4LTFJZQLV
DKWKWJDVE7LCTKBDXCLV3673QV3AJAA5C7XT2CF65FPIFPY76UC4YTTNKN3MHY3BNHD3LOJKGB2U
QLAOLZZ625OCRCS77GX6AJQJFH5I2LPXJLJWEGVKJN6J7MDOEHPBU37ZZQMLLVYPFUT57ADBB67C
5Q
END-RANK245-7X7X7-ZLIB-BASE32
\end{Verbatim}

\newpage
This Python program reads a copy of the scheme above from standard
input and verifies that it represents the $7\times7\times7$ matrix
multiplication tensor over $\mathbb F_2$.

\begin{Verbatim}[fontsize=\scriptsize]
import base64
import sys
import zlib

BEGIN = "BEGIN-RANK245-7X7X7-ZLIB-BASE32"
END = "END-RANK245-7X7X7-ZLIB-BASE32"
N = 7
RANK = 245

def support_bits(mask):
    return [b for b in range(N * N) if mask & (1 << b)]

text = sys.stdin.read()
payload = text.split(BEGIN, 1)[1].split(END, 1)[0]
payload = "".join(payload.split())
payload += "=" * (-len(payload) % 8)
data = zlib.decompress(base64.b32decode(payload))
if len(data) != RANK * 3 * N:
    raise ValueError("unexpected decoded length")
blocks = set()
for o in range(0, len(data), 21):
    blocks.add(data[o : o + 21])
if len(blocks) != RANK:
    raise ValueError("duplicate rank-one terms")

scheme = []
for o in range(0, len(data), 21):
    masks = [
        int.from_bytes(data[o + p : o + p + 7], "little")
        for p in (0, 7, 14)
    ]
    if any(mask >> (N * N) for mask in masks):
        raise ValueError("nonzero padding bits")
    supports = [support_bits(mask) for mask in masks]
    if any(not s for s in supports):
        raise ValueError("zero factor in rank-one term")
    scheme.append(supports)
if len(scheme) != RANK:
    raise ValueError("wrong number of rank-one terms")

parity = set()
for A, B, C in scheme:
    for a in A:
        for b in B:
            for c in C:
                key = (a, b, c)
                if key in parity:
                    parity.remove(key)
                else:
                    parity.add(key)

expected = set()
for i in range(N):
    for j in range(N):
        for k in range(N):
            expected.add(
                (i * N + j, j * N + k, i * N + k))

if parity != expected:
    raise ValueError(
        f"{len(expected - parity)} missing, "
        f"{len(parity - expected)} extra")

print("verified rank-245 decomposition of M^(7) over F_2")
\end{Verbatim}

\section{Rank Reduction of Generalized Flips}\label{app:rank-reduction}

We now formalize the best possible rank reduction that can be obtained from a generalized flip on a group of summands sharing one tensor factor. Suppose $T = \sum_{i=1}^{N} A\otimes B_i\otimes C_i$ is a group of $N$ summands that share the first factor. Consider $E=\mathbb F_2^N$ with the standard basis $e_1, \ldots, e_N,$ and define linear maps $B:E \to V_B$, $e_i \mapsto B_i$ and
$C: E \to V_C$, $e_i \mapsto C_i$.
Equivalently, the ``non-shared'' part of the group is the order-two tensor
$T_{BC} = \sum_{i=1}^N B_i \otimes C_i \in V_B \otimes V_C$.
A generalized flip chooses the basis $g_1, \ldots, g_N$ of $E$ and the dual basis $h_1, \ldots, h_N$ of $E^\ast.$ It then rewrites
\[ \sum_{i=1}^{N} A\otimes B_i\otimes C_i \quad\text{as}\quad \sum_{j=1}^{N} A\otimes B(g_j)\otimes C(h_j), \] where we identify $C$ with the induced map $E^\ast \to V_C$ given by $h\mapsto \sum_{i=1}^{N} h(e_i)C_i$.
Equivalently, if $G\in \mathrm{GL}_N(\mathbb{F}_2)$ has columns $g_1, \ldots, g_N,$ this is the transformation $B,C \longmapsto BG,\;G^{-1}C$.

\begin{theorem}[Optimal Generalized Flip Reduction]
\label{thm:opt-gen-flip-reduction}
Let $T = \sum_{i=1}^{N} A\otimes B_i\otimes C_i$ be a group of $N$ summands sharing the first tensor factor. Among all generalized flips $B,C \longmapsto BG,\;G^{-1}C$ with $G\in \mathrm{GL}_N(\mathbb{F}_2)$, the minimum possible number of non-zero surviving summands is $\operatorname{rank}(T_{BC})$, where $T_{BC}=\sum_{i=1}^{N} B_i\otimes C_i$ is viewed as an order-two tensor, or equivalently as a linear map $V_C^\ast \to V_B.$ Therefore the maximum number of summands that can be deleted by a generalized flip on this group with a shared factor is just $N-\operatorname{rank}(T_{BC})$. Analogous statements hold for groups sharing the $B$ or $C$ factor.
\end{theorem}

\begin{proof}
    Let $U=\ker B\subseteq E$. A transformed $B$-factor is zero precisely when its basis vector lies in $U.$ Next define $W=\{h\in E^\ast : C(h)=0\}$. Thus a transformed $C$-factor is zero precisely when the corresponding dual basis vector lies in $W$. Let $L=W^\perp\subseteq E$. Equivalently, $L$ is the image of the dual map $C^\ast: V_C^\ast \to E$. The order-two tensor $T_{BC}=\sum_{i=1}^{N} B_i\otimes C_i$ corresponds to the linear map $V_C^\ast \longrightarrow V_B$, $\varphi \longmapsto \sum_{i=1}^{N} \varphi(C_i)B_i$. This map factors as \[ V_C^\ast \xrightarrow{C^\ast} L \xrightarrow{B|_L} V_B. \] Since $C^\ast$ has image $L,$ we have $\operatorname{rank}(T_{BC}) = \operatorname{rank}(B|_L) = \dim L-\dim(U\cap L)$.

    We now construct a generalized flip with exactly this many surviving summands. Choose a basis of $E$ adapted to the pair $U, L$ as follows: first choose a basis of $U \cap L,$ extend it to a basis of $L,$ and then extend it to a basis of all of $E.$ Let $g_1, \ldots, g_N$ be the resulting basis, and let $h_1, \ldots, h_N$ be the dual basis.

    If $g_j \in U \cap L,$ then $B(g_j)=0$, so the $j$th transformed summand is deleted. If $g_j \not \in L,$ then the dual vector $h_j$ annihilates $L,$ hence lies in $L^\perp=W,$ and so $C(h_j)=0,$ and the $j$th transformed summand is also deleted.

    The only surviving summands correspond to those basis vectors in $L/(U\cap L)$. The number of surviving summands is $\dim L-\dim(U\cap L) = \operatorname{rank}(T_{BC})$.

    It remains to show optimality. Every generalized flip rewrites $T_{BC}$ as a sum of surviving non-zero terms $B(g_j)\otimes C(h_j)$. Thus, if a generalized flip leaves $s$ non-zero summands, then $T_{BC}$ has an order-two tensor decomposition using $s$ simple tensors. Hence $s\geq \operatorname{rank}(T_{BC})$.

    The construction above achieves equality, so the minimum possible number of survivors is exactly $\operatorname{rank}(T_{BC}).$
\end{proof}

\Cref{thm:opt-gen-flip-reduction} suggests a natural implementation question: can the optimal local rank reduction be computed efficiently enough to improve a random-walk search, especially for groups larger than a single warp can enumerate by subset tests? Making this optimal reduction efficient on the GPU is one of the directions we discuss in~\Cref{sec:future-work}.


\section{A Gap in the Previous Connectivity Proof}\label{app:gap}

Connectivity over $\mathbb{F}_2$ was first proved by Arai et al.~\cite{arai2024adaptive}, by rewriting $S$ into $S'$ in place. Both schemes are ordered, and for each $k$, the $k$-th summand of the current scheme is transformed into the $k$-th summand of $S'$ by repeated plus transformations. This also generates additional summands, which are then appended beyond position $k$ and will eventually be deleted via reduction transformations. For the procedure to work correctly, stage $k$ must preserve summands at positions $1, \ldots, k-1$, since those have already been matched. This invariant is never stated, and the transformations the proof uses do not preserve it.

Nothing in their proof guarantees that the pluses used do not interfere with the summands that need to be preserved.

To better illustrate the gap in the proof, here is a counterexample. Let $T$ be any multiset of summands for which $T \uplus \{A \otimes B \otimes C\}$ is a scheme, let $P \neq Q$ with $P$, $Q$, $A{+}P$, $A{+}Q$ all nonzero, and let
\[
\begin{aligned}
S  &= T \uplus \{(A{+}P)\otimes B \otimes C,\; P\otimes B \otimes C\},\\
S' &= T \uplus \{(A{+}Q)\otimes B \otimes C,\; Q\otimes B \otimes C\}.
\end{aligned}
\]
Both are schemes of the same rank, since $(A{+}P)\otimes B\otimes C + P\otimes B\otimes C = A\otimes B\otimes C$ over $\mathbb{F}_2$, and likewise for $Q$. Order both with $T$ first. The stages covering $T$ have nothing to do, so the proof reaches the two remaining summands and must turn $(A{+}P)\otimes B\otimes C$ into $(A{+}Q)\otimes B\otimes C$. Its first plus needs a partner differing from the target in all three components. The one unmatched summand available, $P \otimes B \otimes C$, shares both $B$ and $C$ with the target and is not eligible. Every eligible summand for a plus therefore lies in $T$, but applying a plus with an element in $T$ would modify that element of $T$, rendering false the invariant that at stage $k$, summands at positions $1,\ldots,k-1$ are preserved. Arai et al.\ do not demonstrate how to repair this invariant, thus their proof is incomplete as stated.

Our construction is not exposed to this issue, because it never rewrites $S$.~\Cref{lemma:add-duplicate-sum} restores both summands it consumes, and~\Cref{lem:add-arbitrary} adds two copies of an arbitrary summand onto a scheme it leaves intact. Adding every summand of $S'$ in any order therefore cannot disturb earlier work, and~\Cref{lemma:zero-subset} removes the remainder.

\end{document}